\documentclass[10pt,twocolumn,twoside]{IEEEtran}

\usepackage{setspace,multirow,tikz,etoolbox}
\usetikzlibrary{arrows,automata}
\usepackage{lipsum, enumitem}

\usepackage{blkarray}
\usepackage[T1]{fontenc}
\usepackage[latin9]{inputenc}
\usepackage{amsthm}
\usepackage{amsmath}
\usepackage{graphicx}
\usepackage{amssymb}
\usepackage{url}
\usepackage{verbatim}
\usepackage[export]{adjustbox}
\usepackage{framed,caption}
\usepackage{multirow}
\usepackage[ruled,vlined]{algorithm2e}

\newtoggle{Arxiv}
\toggletrue{Arxiv}

\theoremstyle{plain}
\newtheorem{thm}{Theorem}
\theoremstyle{plain}

\usepackage{cite}\usepackage{amsfonts}\usepackage{times}\usepackage{bm}
\usepackage{amsthm}
\usepackage{subfigure}
\theoremstyle{definition}

\newtheorem{definition}{Definition}

\newtheorem{lemma}{Lemma}

\newtheorem*{lemma*}{Lemma}

\theoremstyle{remark}

\title{{\huge Multi-scale Spectrum Sensing in Small-Cell mm-Wave Cognitive Wireless Networks}}

\author{Nicolo Michelusi, Matthew Nokleby, Urbashi Mitra, and Robert Calderbank
\thanks{This research has been funded in part by the following grants: ONR N00014-15-1-2550, ONR N00014-09-1-0700, NSF CNS-1213128, NSF CCF-1410009, AFOSR FA9550-12-1-0215 and NSF CPS-1446901.}
\thanks{The research of N. Michelusi  has been funded by NSF under grant CNS-1642982, and by DARPA under grant \#108818.}
\thanks{N. Michelusi is with the School of Electrical and Computer Engineering, Purdue University. email: michelus@purdue.edu.}
\thanks{M. Nokleby is with the Dept. of Electrical and Computer Engineering, Wayne State University. email:matthew.nokleby@wayne.edu}
\thanks{U. Mitra is with the Dept. of Electrical Engineering, University of Southern California. email: ubli@usc.edu.}
\thanks{R. Calderbank is with the Dept. of Electrical Engineering, Duke University. email: robert.calderbank@duke.edu.}
 \vspace{-7mm}
}

\graphicspath{%
     {./Figures/}
}

\begin{document}
\maketitle
\pagenumbering{gobble}

\begin{abstract}
In this paper, a multi-scale approach to spectrum sensing in cognitive cellular networks is proposed. 
In order to overcome the huge cost incurred in the acquisition of full network state information, 
a hierarchical scheme is proposed, based on which local state estimates are aggregated up the hierarchy to obtain
aggregate state information at multiple scales, which are then sent back to each cell for local decision making.
Thus, each cell obtains fine-grained estimates of the channel occupancies of nearby cells, but coarse-grained estimates of those of distant cells.
The performance of the aggregation scheme is studied in terms of the trade-off between the throughput achievable by secondary users and the interference
generated by the activity of these secondary users to primary users.
In order to account for the irregular structure of interference patterns arising from path loss, shadowing, and blockages, which are especially relevant in millimeter wave networks,
a greedy algorithm is proposed to find a multi-scale aggregation tree to optimize the performance.
It is shown numerically that  this tailored hierarchy outperforms a regular tree construction by  60\%.
\end{abstract}
\vspace{-5mm}

\section{Introduction}
 The recent proliferation of mobile devices has been exponential
in number as well as heterogeneity \cite{CISCO}.
This tremendous increase in demand of wireless services poses severe challenges due to
the finite bandwidth of current systems, and calls for new tools for the design and optimization of \emph{agile} wireless networks \cite{pcast}. 
Cognitive radio \cite{Mitola} has the potential to improve spectral efficiency, by enabling smart terminals (secondary users, SUs) to exploit resource gaps left by legacy, primary users (PUs)~\cite{Peha}.

In this paper, we consider a cognitive cellular network, which comprises a set of PUs, which are licensed to access the spectrum, and a set of SUs, which may access opportunistically any unoccupied spectrum. The network is arranged into cells. In each cell, PUs join and leave the channel at random times, thus the state of each cell is described by a first-order binary Markov process.
 In order to utilize the unoccupied spectrum, the SUs require accurate estimates of spectrum occupancies throughout the cellular network. In principle, 
 the channel occupancies can be sensed locally in each cell and collected at a fusion center; the global network state information collected at the fusion center is then broadcasted
 to each cell for
 local decision making. In practice, however, such centralized estimation can be extremely costly in terms of transmit energy and time. 
 
Due to path loss, shadowing, and blockage, SUs accessing the channel in one cell cause significant interference to nearby PUs, but negligible interference to distant PUs. Therefore, each SU needs precise information about the occupancies of nearby cells, but only coarse information about the occupancies of faraway cells.
Given this intuition, we construct a cellular hierarchy, which is used to \emph{aggregate} channel measurements over the network at multiple scales.
Thus, SUs operating in a given cell have precise knowledge about the local state, \emph{aggregate} knowledge of the states of the cells nearby,
\emph{aggregate and coarser} knowledge of the states of the cells farther away, and so on at multiple scales, reflecting the distance dependent nature of wireless interference.

This paper provides important extensions over \cite{MichelusiGCOM}, wherein we assumed a regular tree for hierarchical spectrum sensing by assuming that interference is regular and isotropic (matched to the hierarchy). Herein, we examine the irregular effects of shadowing and blockage, which are
especially severe at millimeter wave frequencies \cite{Singh,Singh2,Bai15}. As in \cite{MichelusiGCOM}, we tradeoff SU network throughput versus the interference generated by the SUs to the PUs. To overcome combinatorial complexity, we develop a greedy algorithm to determine the best hierarchical aggregation tree matched to the irregular interference patterns of millimeter wave communications. Optimality is defined in terms of the trade-off between SU network throughput and interference to PUs.  As expected, this tailored hierarchy outperforms the regular tree construction in \cite{MichelusiGCOM} by  60\%.  Our methods also apply to sub-6GHz wireless networks and are robust to issues of directionality of interferers and primary receivers.

 Hierarchical estimation was proposed in \cite{nokleby:JSTSP13} in the context of averaging consensus \cite{benezit:IT10}, which is a prototype for distributed, linear estimation schemes. Consensus-based schemes for spectrum estimation have also been proposed in~\cite{li:TVT10,zeng:JSTSP11}.
 In contrast, we focus on a \emph{dynamic} setting.
  A framework  for joint spectrum sensing and scheduling in wireless networks has been proposed in \cite{myTCNC},
 for the case of a single cell. Instead, we consider a network composed of multiple cells.
 
To summarize, the contributions of this paper are as follows.
1) We propose a hierarchical framework for aggregation of channel state information over a wireless network composed of multiple cells,
with a generic interference pattern among cells.
We study the performance of the aggregation scheme in terms of the trade-off between the throughput of SUs and the interference
generated by the activity of the SUs to the PUs.
2) We develop a closed form expression for the belief of the spectrum occupancy vector that shows that this belief is statistically independent across subsets of cells at different levels of the hierarchy,
and uniform within each subset (Theorem \ref{thm1}).  This results greatly facilitates the computation of the expected average long-term reward (Lemma \ref{lem1}); and 
3) we address the optimal design of the hierarchical aggregation tree so as to optimize performance, for a given
interference pattern; due to the combinatorial complexity of this problem, we propose a greedy algorithm based on 
agglomerative clustering \cite[Ch. 14]{friedman:01} (Algorithm \ref{alg:clustering}).

This paper is organized as follows. 
In Sec. \ref{sysmo}, we present the system model.
In Sec. \ref{analysis}, we present the performance analysis, for a given tree and interference pattern.
In Sec. \ref{treedesign}, we address the tree design.
In Sec. \ref{numres}, we present numerical results and, in Sec. \ref{conclu}, we conclude this paper.

\begin{figure}[t]
\centering  
\includegraphics[width=.8\linewidth,trim = 0mm 0mm 0mm 0mm,clip=true]{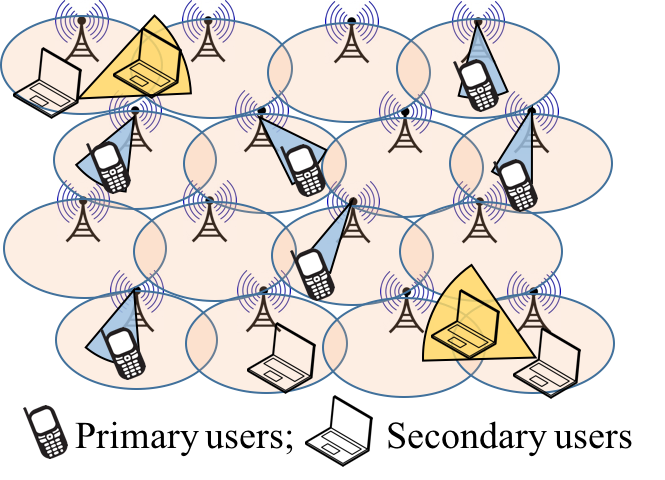}
\caption{System model.}
\label{fig:sysmo}
\vspace{-5mm}
\end{figure}

\vspace{-5mm}
\section{System model}
 \label{sysmo}
We consider a cognitive network, depicted in Fig. \ref{fig:sysmo}, composed of a primary cellular network
with $N_C$ cells, and an opportunistic network of SUs.
Cells are indexed as $1,2,\dots, N_C$.
We denote the set of cells as $\mathcal C\equiv\{1,2,\dots, N_C\}$.
The SUs opportunistically access the spectrum so as to maximize their own throughput, under a constraint on the interference caused to the cellular network.

Let $b_{i,t}\in \{0,1\}$ be the PU spectrum occupancy of cell $i\in\mathcal C$ in slot $t$. That is, $b_{i,t}=1$
if the channel is occupied by PUs in cell $i$ at time $t$, and $b_{i,t}=0$ if it is idle.
We suppose that $\{b_{i,t},t\geq 0,i\in\mathcal C\}$ 
are i.i.d. across cells, and evolve according to a two-state Markov chain, as a result of PUs joining and leaving the network at random times.
We let 
\begin{align}
\label{txprob}
&p\triangleq\mathbb P(b_{i,t+1}=1|b_{i,t}=0),
\
q\triangleq\mathbb P(b_{i,t+1}=0|b_{i,t}=1),
\end{align}
be the transition probability of the Markov chain from "0" to "1" and from "1" to "0", respectively.
Therefore, the steady-state probability that $b_{i,t}$ is occupied is given by
\begin{align}
\pi_B\triangleq\frac{p}{p+q}.
\end{align}
We denote the state of the network in slot $t$ as $\mathbf b_t=(b_{1,t},b_{2,t},\dots,b_{N_C,t})$.

The activity of the SUs is represented by the \emph{SU access decision} $a_{i,t}\in\{0,1\}$, in cell $i$, slot $t$,
where $a_{i,t}=1$ if the SUs operating in cell $i$ access the channel at time $t$, and 
$a_{i,t}=0$ otherwise.
We denote the network-wide SU access decision as $\mathbf a_t=(a_{1,t},a_{2,t},\dots,a_{N_C,t})$ in slot $t$.
The activity of the SUs generate interference to the cellular network.
We denote the interference strength between cells $i$ and $j$ as
$\phi_{i,j}\geq 0$. We assume that interference is symmetric, so that $\phi_{i,j}=\phi_{j,i},\forall i,j\in\mathcal C$.
Note that $\phi_{i,i}$ is the  strength of the interference caused by the SUs in cell $i$ to cell $i$.
We let $\boldsymbol{\Phi}$ be the symmetric interference matrix,
with components $[\boldsymbol{\Phi}]_{i,j}=\phi_{i,j},\forall  i,j\in\mathcal C$.

Given the network state $\mathbf b_t\in\{0,1\}^{N_C}$ and the SU access decision $a_{i,t}\in\{0,1\}$,
we define the local reward for the SUs in cell $i$ as
\begin{align}
\label{locrew}
\!\!\!\!\!r_{i}(a_{i,t},\mathbf b_t) {=} a_{i,t}\left[\rho_I (1{-} b_{i,t}) + \rho_Bb_{i,t}
-\lambda\sum_{j=1}^{N_C} \phi_{i,j}b_{j,t}\right].
\end{align}
The term $a_{i,t}(1- b_{i,t})$ in (\ref{locrew}) 
equals one if and only if the SUs in cell $i$ access the channel when cell $i$ is idle;
$\rho_I\geq 0$ is the instantaneous expected SU throughput accrued in this case.
The term $a_{i,t}b_{i,t}$ in (\ref{locrew}) 
equals one if and only if the SUs in cell $i$ access the channel when cell $i$ is occupied;
$\rho_B$ is the instantaneous expected SU throughput accrued in this case, with $0\leq\rho_B\leq\rho_I$.
Finally, the term
\begin{align*}
a_{i,t}\sum_{j=1}^{N_C} \phi_{i,j}b_{j,t}
\end{align*}
represents the overall interference generated by the
SUs in cell $i$ to the rest of the network, cell $i$ included.
The term $\lambda>0$ is a Lagrangian multiplier which captures the trade-off between
the reward for the SU system and the interference generated to the PUs.

The network reward is defined as the aggregate reward over the entire network,
as a function of the SU access decision $\mathbf a_t$ and network state $\mathbf b_t$,
\begin{align}
R(\mathbf a_t,\mathbf b_t)=\sum_{i\in\mathcal C}r_{i}(a_{i,t},\mathbf b_t).
\end{align}

The SU access decision in cell $i$ is decided based on partial network state information,
denoted as $\pi_{i,t}$ at time $t$, where $\pi_{i,t}(\mathbf b)$ is the belief that 
the network state takes value $\mathbf b_t$ in slot $t$, available to SUs in cell $i$. Given $\pi_{i,t}$,
the SUs in cell $i$ choose $a_{i,t}\in\{0,1\}$ so as to maximize the expected reward $r_{i}(a_{i,t},\pi_{i,t})$,
given by
\begin{align}
\label{expreward}
r_{i}(a_{i,t},\pi_{i,t})\triangleq\sum_{\mathbf b\in\{0,1\}^{N_C}}\pi_{i,t}(\mathbf b)r_{i}(a_{i,t},\mathbf b).
\end{align}
Thus,
\begin{align}
\label{suopt}
a_{i,t}^*=\arg\max_{a\in\{0,1\}}r_{i}(a,\pi_{i,t}),
\end{align}
yielding the optimal expected local reward
\begin{align}
r_{i}^*(\pi_{i,t})&=\max\{r_{i}(0,\pi_{i,t}),r_{i}(1,\pi_{i,t})\}\nonumber
=\max\{0,r_{i}(1,\pi_{i,t})\},
\end{align}
where $r_{i}(0,\pi_{i,t})=0$ from (\ref{locrew}).

Given the belief $\boldsymbol{\pi}_t=(\pi_{1,t},\pi_{2,t},\dots,\pi_{N_C,t})$ across the network,
under the optimal SU access decisions $\mathbf a_{t}^*$ given by (\ref{suopt}), the optimal network reward is thus given by
\begin{align}
\label{netrew}
R^*(\boldsymbol{\pi}_t)=\sum_{i\in\mathcal C}r_{i}^*(\pi_{i,t}).
\end{align}

Using the fact that $r_{i}(a,\mathbf b)\leq \max\{r_{i}(0,\mathbf b),r_{i}(1,\mathbf b)\},\forall i\in\mathcal C,\forall a\in\{0,1\},\forall \mathbf b\in\{0,1\}^{N_C}$,
we obtain the inequality 
\begin{align}
\label{upbound}
R^*(\boldsymbol{\pi}_t)
&\leq
\sum_{\mathbf b\in\{0,1\}^{N_C}}\pi_{i,t}(\mathbf b)\sum_{i\in\mathcal C}\max\{0,r_{i}(1,\mathbf b)\},
\end{align}
\emph{i.e.}, the expected network reward under partial network state information is upper bounded
by the expected network reward obtained when full network state information is provided to the SUs in each cell (perfect knowledge of $\mathbf b_t$).
Thus, the SUs should, possibly, obtain full network state information in order to achieve the best performance.

The belief $\pi_{i,t}$ is computed based on spectrum measurements performed over the network.
Ideally, in order to achieve global network state information and maximize the reward (see (\ref{upbound})),
the SUs in cell $i$ should obtain the local spectrum state $b_{i,t}$, as well as the spectrum state from the rest of the network.
To this end, the SUs in cell $j{\neq}i$ should report the local spectrum state $b_{j,t}$ to the SUs in cell $i$ via information exchange, potentially over multiple hops, for transmitters/receivers far away from each other.
Since this needs to be done over the entire network (\emph{i.e.}, for every pair $(i,j)\in\mathcal C^2$), the associated cost of information exchange may be huge, especially in large networks
composed of a large number of small cells. In order to reduce the cost of acquisition of network state information, in this paper we propose a  
{\em multi-scale} approach to spectrum sensing. To this end, we structure the cellular grid in
a hierarchical structure, defined by a tree of depth $D\geq 1$.
\vspace{-5mm}
\subsection{Tree construction}
Herein, we describe the tree construction.
At level $0$, we have the leaves, represented by the cells $\mathcal C$.
We let $\mathcal C_{0}^{(i)}\equiv\{i\}$ for $i\in\mathcal C$.
At level $1$, let $\mathcal C_{1}^{(k)},k=1,2,\dots,n_{1}$ be a partition of the cells into $n_{1}$ subsets, where $1\leq n_{1}\leq |\mathcal C|$.
We associate a cluster head to each subset $\mathcal C_{1}^{(k)}$;
the set of $n_{1}$ level-$1$ cluster-heads is denoted as $\mathcal H_{1}$.
Hence, $\mathcal C_{1}^{(k)}$  is the set of cells associated to
the level-1 cluster head $k\in\mathcal H_{1}$.

Recursively, at level $L$, let $\mathcal H_{L}$ be the set of nodes defining the level $L$-cluster heads,
 with $L\geq 1$. If  $|\mathcal H_{L}|=1$, then 
 we have defined a tree with depth $D=L$.
 Otherwise, we define a partition of $\mathcal H_{L}$ into
$n_{L+1}$ subsets $\mathcal H_{L}^{(k)},k=1,2,\dots,n_{L+1}$,
where $1\leq n_{L+1}\leq |\mathcal H_{L}|$,
and we associate to each subset a level-$(L+1)$ cluster head;
the set of $n_{L+1}$ level-$(L+1)$ cluster-heads is denoted as $\mathcal H_{L+1}$.
Let $\mathcal C_{L+1}^{(k)},k=1,2,\dots,n_{L+1}$ be the set of cells associated to level-$(L+1)$ cluster head $k\in\mathcal H_{L+1}$.
This is obtained recursively as
\begin{align}
\mathcal C_{L+1}^{(k)}=\bigcup_{m\in\mathcal H_{L}^{(k)}}\mathcal C_{L}^{(m)}.
\end{align}

Let $P_L(i)\in\mathcal H_L$ be the level $L$ parent of cell $i$,
\emph{i.e.}, $P_0(i)=i$, and $P_L(i)=k$ for $L\geq 1$ if and only if $i\in\mathcal C_{L}^{(k)}$, for some $k\in\mathcal H_L$.
We make the following definition.
\begin{definition}We define
the \emph{hierarchical distance} between cells $i\in\mathcal C$ and $j\in\mathcal C$ as
\begin{equation*}
\Lambda(i,j){=}\min\left\{L\geq 0:P_L(i)=P_L(j)\right\}.
\end{equation*}
\end{definition}
In other words, $\Lambda(i,j)$ is the lowest level $L$ such that cells $i$ and $j$ belong to the same cluster at level $L$. 
It follows that $\Lambda(i,i)=0$ and $\Lambda(i,j)=\Lambda(j,i)$, \emph{i.e.}, the hierarchical distance between cell $i$ and itself is $0$,
and it is symmetric.

We let $\mathcal C_{\Lambda}^{(i)}(L)$ be the set of cells at hierarchical distance $L$  from cell $i$.
That is, $\mathcal C_{\Lambda}^{(i)}(0)\equiv\{i\}$, and, for $L>0$,
\begin{align}
\mathcal C_{\Lambda}^{(i)}(L)\equiv\mathcal C_{L}^{(P_L(i))}\setminus\mathcal C_{L-1}^{(P_{L-1}(i))}.
\end{align}
In fact, by the tree construction, $\mathcal C_{L}^{(P_L(i))}$ contains all cells with hierarchical distance (from cell $i$) less (or equal) than $L$.
Thus, $\mathcal C_{\Lambda}^{(i)}(L)$ is obtained by removing from $\mathcal C_{L}^{(P_L(i))}$ all cells  with hierarchical distance less (or equal) than $L-1$, $\mathcal C_{L-1}^{(P_{L-1}(i))}$.

\vspace{-5mm}
\subsection{Hierarchical information exchange over the tree}
In order to collect network state information at multiple scales, the SUs exchange local information
over the tree.
In particular, we propose a scheme in which the SUs  carry out a {\em hierarchical} fusion of local estimates. This fusion is patterned after {\em hierarchical averaging}, a technique for scalar average consensus in wireless networks developed in \cite{nokleby:JSTSP13}.

At the beginning of slot $t$,
at the cell level (level-$0$), the local SUs perform spectrum sensing to estimate the local state.
Thus, the SUs in cell $i$ estimate the local state $b_{i,t}$ as $\hat b_{i,t}^{(i)}\in[0,1]$, representing the belief that the local state takes the value $b_{i,t}=1$, as seen by the SUs operating in cell $i$ (superscript $(i)$).
For simplicity, in this paper we assume that local spectrum sensing is done with no errors, so that
 \begin{align}
 \label{localestimates}
\hat b_{i,t}^{(i)}=b_{i,t},\ \forall i\in\mathcal C.
 \end{align}

Next, these observations are fused up the hierarchy. 
The level $1$ cluster head $m\in\mathcal H_1$ receives the spectrum measurements from
its cluster $\mathcal C_1^{(m)}$, and fuses them as
\begin{align}
S_{m,t}^{(1)}
=\sum_{j\in\mathcal C_1^{(m)}} b_{j,t},\ \forall m\in\mathcal H_1.
\end{align}
This process continues up the hierarchy: the level $L$ cluster head 
$m\in\mathcal H_L$ receives the aggregate spectrum measurements $S_{k,t}^{(L-1)}$ from
the level-$(L-1)$ cluster heads $k\in\mathcal H_{L-1}^{(m)}$ connected to it, and fuses them as 
\begin{align}
S_{m,t}^{(L)}
=\sum_{k\in\mathcal H_{L-1}^{(m)}} S_{k,t}^{(L-1)}
=\sum_{j\in\mathcal C_L^{(m)}} b_{j,t},\ \forall m\in\mathcal H_L.
\end{align}

Eventually, the aggregate spectrum measurements are fused at the unique root of the tree (level $D$) as
\begin{align}
S_{1,t}^{(D)}
=\sum_{k\in\mathcal H_{D-1}^{(1)}} S_{k,t}^{(D-1)}
=\sum_{j\in\mathcal C} b_{j,t},
\end{align}
where we have used the fact that $\mathcal C_D^{(1)}\equiv\mathcal C$.

Upon reaching level $D$,
the appropriate aggregate spectrum measurements are propagated down to the individual cells $i\in\mathcal C$, following the tree.
Thus, at the beginning of slot $t$, the SUs operating in cell $i$ receive 
\begin{align*}
\left\{
\begin{array}{rcl}
S_{P_0(i),t}^{(0)}&=&\sum_{j\in\mathcal C_0^{(P_0(i))}} b_{j,t}=b_{i,t},\\
S_{P_L(i),t}^{(L)}&=&\sum_{j\in\mathcal C_L^{(P_L(i))}} b_{j,t},\ 1\leq L<D,\\
S_{1,t}^{(D)}&=&\sum_{j\in\mathcal C}b_{j,t},
\end{array}\right.
\end{align*}
where we remind that $P_L(i)$ is the level-$L$ parent of cell $i$, and $\mathcal C_L^{(P_L(i))}$ is the set of cells associated to $P_L(i)$.
That is, the SUs operating in cell $i$ receive from the level-$L$ parent the aggregate spectrum measurements over $\mathcal C_L^{(P_L(i))}$.
From this set of measurements, one can compute
\begin{align}
\label{sigmadef}
\left\{
\begin{array}{lcl}
\sigma_{i,t}^{(0)}&\triangleq&b_{i,t},\\
\sigma_{i,t}^{(L)}&\triangleq& S_{P_L(i),t}^{(L)}-S_{P_{L-1}(i),t}^{(L-1)},\ 1\leq L\leq D.
\end{array}\right.
\end{align}
Note that $\sigma_{i,t}^{(L)}$ is the aggregate spectrum measurement of the cells at hierarchical distance $L$ from cell $i$,
\begin{align}
\label{sigmadef2}
\sigma_{i,t}^{(L)}\triangleq\sum_{j\in\mathcal C_{\Lambda}^{(i)}(L)}b_{j,t},\ \forall\ 0\leq L\leq D.
\end{align}
Thus, the SUs in  cell $i$ receive the set of aggregate spectrum measurements at
multiple scales corresponding to
 different hierarchical distances.
Importantly, they know only the aggregate spectrum measurements, but not the specific values of $b_{j,t},\forall j\neq i$.
These aggregate spectrum measurements are used to update the belief $\pi_{i,t}$ in the next section.
\vspace{-4mm}
\section{Analysis}
 \label{analysis}
 The SUs in cell $i$ update the belief $\pi_{i,t}$ based on past and present spectrum measurements
$\boldsymbol{\sigma}_{i,\tau}=(\sigma_{i,\tau}^{(0)},\sigma_{i,\tau}^{(1)},\dots,\sigma_{i,\tau}^{(D)}),\forall 0\leq\tau\leq t$.
The form of $\pi_{i,t}$ is provided by the following Theorem.
\begin{thm}
\label{thm1} 
Given $\boldsymbol{\sigma}_{i,t}=(o_0,o_1,\dots,o_D)$,
\begin{align}
\label{eq1}
\pi_{i,t}(\mathbf b)=\prod_{l=0}^D\mathbb P(b_{j,t}=b_j,\forall j\in\mathcal C_{\Lambda}^{(i)}(L)|\sigma_{i,t}^{(L)}=o_L),
\end{align}
independent of $\boldsymbol{\sigma}_{i,\tau},\forall \tau<t$,
where 
\begin{align}
\label{eq2}
&\mathbb P(b_{j,t}=b_j,\forall j\in\mathcal C_{\Lambda}^{(i)}(L)|\sigma_{i,t}^{(L)}=o_L)
\\&
=
\chi\left(\sum_{j\in\mathcal C_{\Lambda}^{(i)}(L)}b_j=o_L\right)
\frac{o_L!\left(|\mathcal C_{\Lambda}^{(i)}(L)|-o_L\right)!}{|\mathcal C_{\Lambda}^{(i)}(L)|!},\nonumber
\end{align}
where $\chi(\cdot)$ is the indicator function. \hfill\qed
\end{thm}
\begin{proof}
 \iftoggle{Arxiv}{%
See the Appendix.
}{
Due to space constraints, the proof is in \cite{Journ_version}.
}
\end{proof}

From Equation (\ref{eq1}), it follows that the belief $\pi_{i,t}$ is statistically independent 
across the subsets of cells at different hierarchical distances from cell $i$; this result follows from the fact that
 $\{b_{i,t},t\geq 0,i\in\mathcal C\}$ are i.i.d. across cells.
 Additionally, since $\sum_{j\in\mathcal C_{\Lambda}^{(i)}(L)}b_{j,t}=o_L$
 (as a result of state aggregation at hierarchical distance $L$)
  and $b_{j,t}\in\{0,1\}$,
 there are $\left(\begin{array}{c}|\mathcal C_{\Lambda}^{(i)}(L)|\\o_L\end{array}\right)$
 possible combinations of $\{b_{j,t},j\in\mathcal C_{\Lambda}^{(i)}(L)\}$;
 equation (\ref{eq2}) states that these combinations are uniformly distributed,
as a result of the i.i.d. assumption.  

Importantly, $\pi_{i,t}$ is independent of past measurements but solely depends on the 
current one $\boldsymbol{\sigma}_{i,t}$. In fact, spectrum occupancies $b_{j,t}$ are 
identically distributed across cells.
 
 We can use Theorem \ref{thm1} to compute the expected reward in cell $i$, given by (\ref{expreward}).
Using (\ref{locrew}), we obtain
\begin{align}
\label{rewx}
&r_{i}(a_{i,t},\pi_{i,t})=
\rho_I a_{i,t}(1- \mathbb P(b_{i,t}=1|\pi_{i,t})) 
\nonumber\\&
\nonumber
\quad+ \rho_Ba_{i,t}\mathbb P(b_{i,t}=1|\pi_{i,t})
-\lambda a_{i,t}\sum_{j=1}^{N_C} \phi_{i,j}\mathbb P(b_{j,t}=1|\pi_{i,t})
\\&
=
\rho_I a_{i,t}(1- \mathbb P(b_{i,t}=1|\pi_{i,t})) + \rho_Ba_{i,t}\mathbb P(b_{i,t}=1|\pi_{i,t})
\nonumber
\\&
\quad-\lambda a_{i,t}\sum_{L=0}^D\sum_{j\in\mathcal C_{\Lambda}^{(i)}(L)} \phi_{i,j}\mathbb P(b_{j,t}=1|\pi_{i,t}).
\end{align}
In the last step above, we have partitioned the set of cells $\mathcal C$ into the 
subsets corresponding to hierarchical distances $L=0,1,\dots,D$ from cell $i$.
Now, using (\ref{eq2}) in Theorem \ref{thm1}, we obtain, for all $0\leq L\leq D$,
for all $\forall j\in\mathcal C_{\Lambda}^{(i)}(L)$,
\begin{align}
&\mathbb P(b_{j,t}=1|\pi_{i,t})=0,\ \text{if }o_L=0,
\\
&\mathbb P(b_{j,t}=1|\pi_{i,t})=
\frac{o_L!\left(|\mathcal C_{\Lambda}^{(i)}(L)|-o_L\right)!}{|\mathcal C_{\Lambda}^{(i)}(L)|!}
\left(\begin{array}{c}|\mathcal C_{\Lambda}^{(i)}(L)|-1\\o_L-1\end{array}\right),
\nonumber
\\&\qquad\qquad\qquad\qquad\qquad\qquad
\text{if }o_L>0,
\end{align}
since there are $\left(\begin{array}{c}|\mathcal C_{\Lambda}^{(i)}(L)|-1\\o_L-1\end{array}\right)$
combinations such that $b_{j,t}=1$, given that $\sigma_{i,t}^{(L)}=o_L$.
Solving, we obtain
\begin{align}
\mathbb P(b_{j,t}=1|\pi_{i,t})=\frac{o_L}{|\mathcal C_{\Lambda}^{(i)}(L)|}.
\end{align}
Thus, substituting in (\ref{rewx}), and letting
\begin{align}
\Phi_i(L)\triangleq\sum_{j\in\mathcal C_{\Lambda}^{(i)}(L)} \phi_{i,j}
\end{align}
be the total interference generated by the SUs in cell $i$ to the cells at hierarchical distance $L$ from cell $i$,
we can  finally rewrite
\begin{align}
\nonumber
r_{i}(a_{i,t},\boldsymbol{\sigma}_{i,t})=&
\rho_I a_{i,t}(1- \sigma_{i,t}^{(0)}) + \rho_Ba_{i,t}\sigma_{i,t}^{(0)}
\\&
-\lambda a_{i,t}\sum_{L=1}^D\frac{\sigma_{i,t}^{(L)}}{|\mathcal C_{\Lambda}^{(i)}(L)|}
\Phi_i(L),
\end{align}
where, for convenience, we have expressed the dependence of $r_{i}(\cdot)$ on 
$\boldsymbol{\sigma}_{i,t}$, rather than on $\pi_{i,t}$.
Thus, the network reward (\ref{netrew}) is given by
\begin{align}
R^*(\boldsymbol{\Sigma}_t)=\sum_{i\in\mathcal C}\max\left\{0,r_{i}(1,\boldsymbol{\sigma}_{i,t})\right\},
\end{align}
where we have defined $\boldsymbol{\Sigma}_t=[\boldsymbol{\sigma}_{1,t},\boldsymbol{\sigma}_{2,t},\dots,\boldsymbol{\sigma}_{N_C,t}]$,
and, for convenience, we have expressed the dependence of $R^*(\cdot)$ on 
$\boldsymbol{\Sigma}_{t}$, rather than on $\boldsymbol{\pi}_t$.

\subsection{Average long-term performance evaluation}
We are interested in evaluating the average long-term performance of the hierarchical estimation scheme, that is
\begin{align}
\bar R=\lim_{T\to\infty}\frac{1}{T}\mathbb E\left[\sum_{t=0}^{T-1}R^*(\boldsymbol{\Sigma}_t)\right],
\end{align}
where the expectation is computed with respect to the sequence $\{\boldsymbol{\Sigma}_t,t\geq 0\}$.
We have the following result.
\begin{lemma}
\label{lem1}
The average long-term network reward is given by
\begin{align}
&\bar R=
\sum_{i\in\mathcal C}\sum_{o_0\in\{0,1\}}
\mathcal B(o_0;1)
\sum_{o_1=0}^{|\mathcal C_{\Lambda}^{(i)}(1)|}\mathcal B\left(o_1;|\mathcal C_{\Lambda}^{(i)}(1)|\right)\dots \\&
\dots\sum_{o_L=0}^{|\mathcal C_{\Lambda}^{(i)}(L)|}\mathcal B\left(o_L;|\mathcal C_{\Lambda}^{(i)}(L)|\right)
\dots\sum_{o_D=0}^{|\mathcal C_{\Lambda}^{(i)}(D)|}
\mathcal B\left(o_D;|\mathcal C_{\Lambda}^{(i)}(D)|\right)
\nonumber
\\&
\times\max\left\{0,\rho_I(1- o_0) + \rho_Bo_0-\lambda\sum_{L=1}^D\frac{o_L}{|\mathcal C_{\Lambda}^{(i)}(L)|}\Phi_i(L)
\right\},
\nonumber
\end{align}
where $\mathcal B\left(\cdot;|\mathcal C_{\Lambda}^{(i)}(L)|\right)$ is the binomial distribution  with $|\mathcal C_{\Lambda}^{(i)}(L)|$ trials and occupancy probability $\pi_B$,
\begin{align}
\mathcal B\left(o_L;|\mathcal C_{\Lambda}^{(i)}(L)|\right)
{=}
\left(\!\!\!\begin{array}{c}|\mathcal C_{\Lambda}^{(i)}(L)|\\o_L\end{array}\!\!\!\right)
\pi_B^{o_L}(1-\pi_B)^{|\mathcal C_{\Lambda}^{(i)}(L)|-o_L}.\qed
\nonumber
\end{align}
\end{lemma}
In fact, since channel occupancy states are i.i.d. across cells, at steady-state, the number of cells 
occupied
within any subset $\tilde{\mathcal C}\subseteq \mathcal C$ is a binomial random variable
with $|\tilde{\mathcal C}|$ trials (the number of cells in the set)
and  occupancy probability  $\pi_B$ (the steady-state probability that one cell is occupied).
The result then follows by applying this argument to the hierarchical aggregation scheme.

Note that $\bar R$ depends on the structure of the tree employed for network state information exchange.
In the next section, we present an algorithm to design the tree so as to maximize the network reward $\bar R$.

Using (\ref{upbound}), we  can compare the network reward $\bar R$ with the upper bound computed under the assumption of
full network state information at each cell, given by
\begin{align}
\label{upboundavg}
\nonumber
&\bar R_{\mathrm{up}}=
\sum_{i\in\mathcal C}
\sum_{\mathbf b\in\{0,1\}^{N_C}}
\pi_B^{\sum_i b_i}(1-\pi_B)^{N_C-\sum_i b_i}\
\\&
\times
\max\left\{0,\rho_I(1- b_{i}) + \rho_Bb_{i}-\lambda\sum_{j=1}^{N_C} \phi_{i,j}b_{j}\right\}.
\end{align}

\section{Tree Design}
 \label{treedesign}
The reward of the network depends crucially on the tree employed for information exchange. Optimizing the network reward over the set of all possible trees is a combinatorial problem with high complexity. Instead, we employ methods from hierarchical clustering to build a tree.  Hierarchical clustering is well-studied (see, e.g. \cite[Ch. 14]{friedman:01}), with two main approaches: {\em divisive} clustering, in which a tree is built by successively splitting larger clusters;
and {\em agglomerative} clustering, in which a tree is built by successively combining smaller clusters.
Our algorithm is based on the latter.

\begin{algorithm}
\SetKwData{Left}{left}\SetKwData{This}{this}\SetKwData{Up}{up}
\SetKwFunction{Union}{Union}\SetKwFunction{FindCompress}{FindCompress}
\SetKwInOut{Input}{input}\SetKwInOut{Output}{output}
\Input{Cells $\mathcal{C}$, interference matrix $\boldsymbol{\Phi}$
}
\Output{A hierarchy of clusters $\mathcal{C}_L^{(k)}$, $k\in\mathcal H_L$, $L=1,2,\dots, D$}
\BlankLine
\lForEach{cell $i \in \mathcal{C}$}{$\mathcal{C}_0^{(i)} \leftarrow \{i\}$}
initialize $L \leftarrow 0$ \;
\BlankLine
\tcp{if more than one cluster head, continue} 
\While{$|\mathcal H_L|>1$}{
\tcp{make an empty list of next level cluster heads} 
	$\mathcal{H}_{L+1} \leftarrow \emptyset$\; 
	\tcp{cluster head counter}
	$k_{next}\leftarrow 1$\;  
	\tcp{make a list of unpaired cluster heads at the current level} 
	$\mathcal{H}_L^{unpaired} \leftarrow \mathcal{H}_L$\; 
	\While{$|\mathcal{H}_L^{unpaired}|>0$}{
	\If{$|\mathcal{H}_L^{unpaired}| = 1$}{
		\tcp{no unpaired neighbors, ``pair'' with self}
		$k\in\mathcal{H}_L^{unpaired}$ \;
		$\mathcal{H}_{L+1} \leftarrow \mathcal{H}_{L+1}\cup\{k_{next}\}$ \; 
		$\mathcal{C}_{L+1}^{(k_{next})} \leftarrow \mathcal{C}_{L}^{(k)}$ \; 
		$k_{next} \leftarrow k_{next}+1$ \; 
		$\mathcal{H}_L^{unpaired} \leftarrow \mathcal{H}_L^{unpaired}\setminus \{k\}$ \;		
	}
	\Else{
	\tcp{find unpaired cluster with max similarity}
	$(k,k^*) \leftarrow \underset{k,k^\prime \in \mathcal{H}_L^{unpaired},k\neq k^\prime}{\arg\max} \gamma_L(k,k^\prime)$ \;
	$\mathcal{H}_{L+1} \leftarrow \mathcal{H}_{L+1}\cup\{k_{next}\}$ \; 
		$\mathcal{C}_{L+1}^{(k_{next})} \leftarrow \mathcal{C}_{L}^{(k)}\cup\mathcal{C}_{L}^{(k^*)}$ \; 
		$k_{next} \leftarrow k_{next}+1$ \; 
			\tcp{remove paired clusters from list}
		$\mathcal{H}_L^{unpaired} \leftarrow \mathcal{H}_L^{unpaired}\setminus \{k,k^*\}$ \;		
	}
	}
	$L \leftarrow L+1$ \;
}
\caption{Agglomerative Hierarchy Construction}\label{alg:clustering}
\end{algorithm}
\DecMargin{1em}

Agglomerative clustering requires a similarity metric between clusters; at each round, similar clusters are aggregated. Our goal in designing a tree-based approach to spectrum sensing is to prioritize information that nodes can use to limit the interference they generate to other cells. Therefore, we want to aggregate cells together with high potential for interference. To this end, we define the similarity between level-$L$ clusters $k_1,k_2\in\mathcal H_L$ as
\begin{align}
\label{similarity}
	\gamma_L(k_1,k_2) = \sum_{i \in \mathcal{C}_L^{(k_1)}}\sum_{j \in \mathcal{C}_L^{(k_2)}} \phi_{i,j},
\end{align}
or the sum of inter-cluster interference strengths. 

The algorithm proceeds as shown in Algorithm \ref{alg:clustering}. We initialize it with the $N_C$ leaves $\mathcal{C}_0^{(i)}=\{i\},i=1,2,\dots,N_C$.
 Then, at each level $L$, we iterate over all of the clusters, pairing each one with the cluster with which it has the most interference
 (this can be done in order of pairs with maximum similarity (\ref{similarity})). 
  This forms the set of level $L+1$ clusters. If the number of clusters at level $L$ happens to be odd, one cluster may not be paired, in which case it forms its own cluster at level $L+1$. The algorithm continues until the cluster $\mathcal{C}_{L}^{(1)}$ contains the entire network, \emph{i.e.}, a tree is formed.

Agglomerative clustering has complexity $O(N_C^2 \log(N_C))$, where the $N_C^2$ term owes to searching over all pairs of clusters.

 \section{Numerical Results}
  \label{numres}
 In this section, we provide numerical results.
 We consider a $4\times 4$ cells network. We set the parameters as follows: $\rho_I{=}1$, $\rho_B{=}0$, $\lambda{=}1$, $p{=}q{=}0.1$.
 We use the following interference model between a pair of cells (assuming there is no blockage between them):
 \begin{align}
\left\{
\begin{array}{l}
 \phi_{i,j}=\left\|\mathbf p(i)-\mathbf p(j)\right\|^{-\alpha},\ i\neq j,
\\
 \phi_{i,i}=1,
 \end{array}\right.
 \end{align}
 where $\mathbf p(i)$ is the position of cell $i$, $\left\|\mathbf p(i)-\mathbf p(j)\right\|$ is the distance between cells $i$ and $j$,
and  $\alpha=2$ is the pathloss exponent.

\begin{figure}[t]
\centering  
\includegraphics[width=\linewidth,trim = 2mm 0mm 13mm 8mm,clip=true]{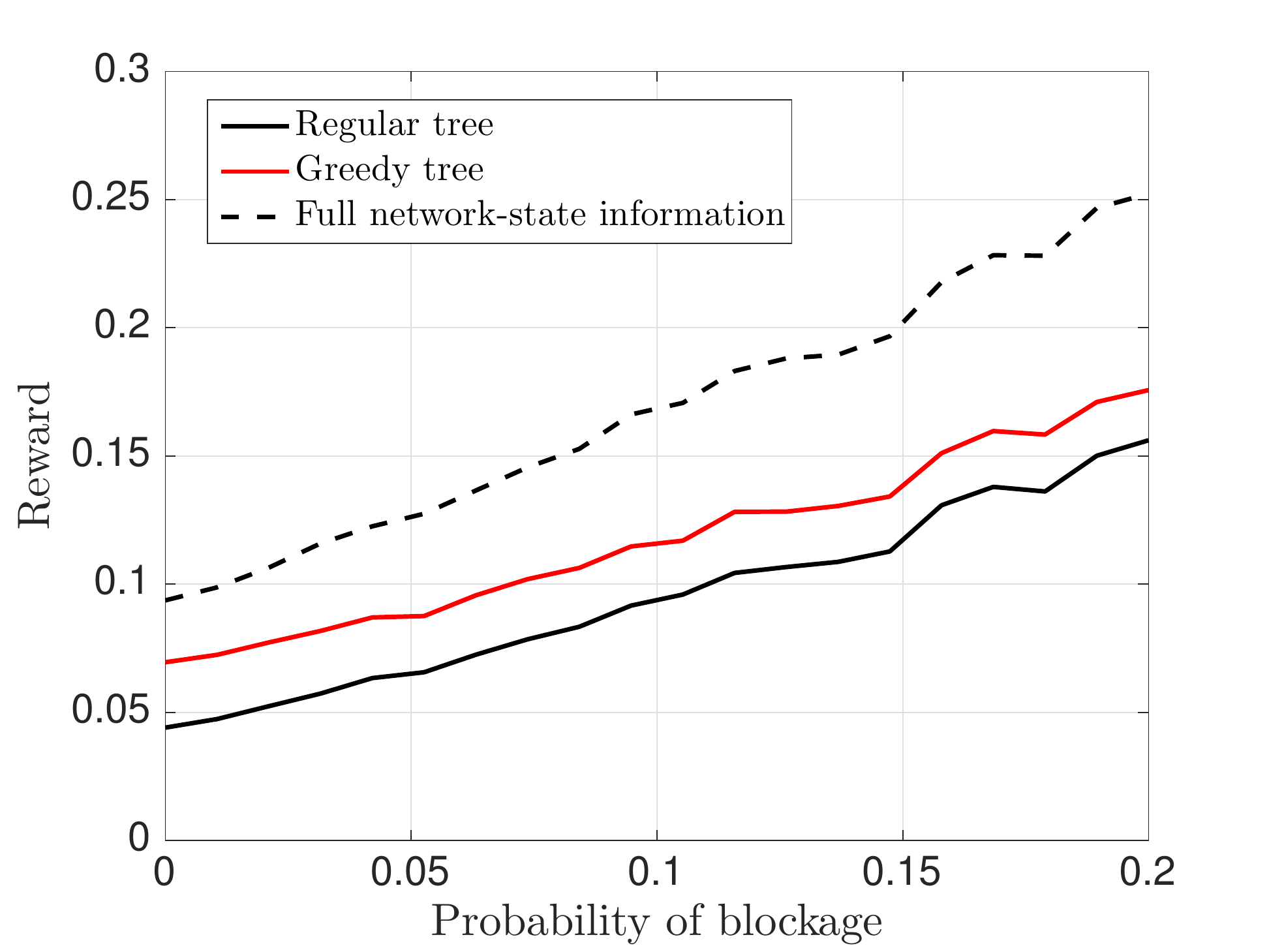}
\caption{Average long-term network reward as a function of the probability of blockage $p_{\mathrm{block}}$.}
\label{fig:simres}
\vspace{-5mm}
\end{figure}

 We define random "walls" between cell boundaries, i.i.d. with probability $p_{\mathrm{block}}\in[0,1]$.
 If a wall is present, then all the cells separated by it experience mutual blockage; thus, if cells $i$ and $j$ are separated by a wall,
 then $\phi_{i,j}=\phi_{j,i}=0$. The blockage topology is generated randomly, for a given blockage  probability $p_{\mathrm{block}}$,
 and a sample average of the performance is computed over $200$ independent trials.
 
In Fig. \ref{fig:simres}, we plot the curve of the network reward as a function of the blockage probability  $p_{\mathrm{block}}$, for different schemes:
\begin{itemize}
\item a scheme in which a \emph{regular tree} is used for state information aggregation. In this case, neighboring cells and clusters are paired together, in order, independently of the interference pattern.
This scheme is similar to \cite{MichelusiGCOM};
\item a scheme in which the tree is generated with Algorithm~\ref{alg:clustering}, by leveraging the specific structure of interference;
\item an upper bound in which the reward is computed under full network state information at each cell, given by (\ref{upboundavg}).
This is computed via Monte Carlo simulation over 5000 independent realizations of $\mathbf b_t$ (at steady-state).
\end{itemize}
We notice that the best performance is obtained under full network state information available at each cell. This is because each cell can leverage the most refined information on
the interference pattern. However, this comes at a huge cost of propagating network state information over the network.
 In contrast, the cost of acquisition of state information can be significantly reduced using aggregation, at the cost of some performance degradation.
Remarkably, by using the greedily optimized algorithm for information aggregation, the performance improves by up to 60\% with respect to a scheme that uses a regular tree.
 In fact, the greedily optimized algorithm leverages the specific structure of interference over the network.

\section{Conclusions}
 \label{conclu}
In this paper, we have proposed a multi-scale approach to spectrum sensing in cognitive cellular networks. 
To reduce the cost of acquisition of network state information, 
we have proposed a hierarchical scheme, that makes it possible to obtain aggregate 
state information at multiple scales, at each cell.
We have studied analytically the performance 
 of the aggregation scheme in terms of the trade-off between the throughput achievable by secondary users and the interference
generated by the activity of these secondary users to primary users.
We have proposed a greedy algorithm to find a multi-scale aggregation tree, matched to the structure of interference, to optimize the performance.
Finally, we have shown performance improvement up to 60\% using a greedily optimized tree, compared to a regular one.

 \iftoggle{Arxiv}{%
\appendices
\section*{Appendix: Proof of Theorem \ref{thm1}}
\begin{proof}
We prove (\ref{eq1}) and (\ref{eq2}) by induction on $t$.
At time $t=0$, given $\boldsymbol{\sigma}_{i,0}=(o_0,o_1,\dots,o_D)$,
using Bayes' rule we obtain
\begin{align}
&\pi_{i,0}(\mathbf b)=\mathbb P(b_{j,0}=b_j,\forall j|\boldsymbol{\sigma}_{i,0}=(o_0,o_1,\dots,o_D))
\\&
=
\frac{\mathbb P(\boldsymbol{\sigma}_{i,0}=(o_0,o_1,\dots,o_D)|b_{j,0}=b_j,\forall j)\mathbb P(b_{j,0}=b_j,\forall j)}
{
\!\!\!\!\!\!\!\!\!\underset{\tilde{\mathbf b}\in\{0,1\}^{N_C}}{\sum}\!\!\!\!\!\!\!\mathbb P(\boldsymbol{\sigma}_{i,0}{=}(o_0,o_1,\dots,o_D)|b_{j,0}=\tilde b_j,\forall j)\mathbb P(b_{j,0}=\tilde b_j,\forall j)
}.
\nonumber
\end{align}
Then, noticing that  $\sigma_{i,0}^{(L)}$ in (\ref{sigmadef2})
is only a function of $b_{j,0},\forall j\in\mathcal C_{\Lambda}^{(i)}(L)$, but is independent of 
$b_{j,0},j\notin \mathcal C_{\Lambda}^{(i)}(L)$,
and since $b_{j,0}$ is statistically independent across cells $j$,
we obtain
\begin{align}
\label{wdsd}
&\pi_{i,0}(\mathbf b)
\\&
=
\frac{\left[
\begin{array}{c}
\prod_{L=0}^D\mathbb P(\sigma_{i,0}^{(L)}=o_L|b_{j,0}=b_j,\forall j\in\mathcal C_{\Lambda}^{(i)}(L))\\
\times\mathbb P(b_{j,0}=b_j,\forall j\in\mathcal C_{\Lambda}^{(i)}(L))
\end{array}\right]
}
{
\!\!\!\!\!\!\!\!\!\underset{\tilde{\mathbf b}\in\{0,1\}^{N_C}}{\sum}\!
\left[
\begin{array}{c}
\prod_{L=0}^D\mathbb P(\sigma_{i,0}^{(L)}=o_L|b_{j,0}=\tilde b_j,\forall j\in\mathcal C_{\Lambda}^{(i)}(L))\\\times\mathbb P(b_{j,0}=\tilde b_j,\forall j\in\mathcal C_{\Lambda}^{(i)}(L))
\end{array}\right]
}.
\nonumber
\end{align}
Using the fact that $\{\mathcal C_{\Lambda}^{(i)}(L),L=0,1,\dots, D\}$ define a partition of $\{1,2,\dots, N_C\}$,
we have that 
\begin{align}
\label{fLprop}
&\sum_{\tilde{\mathbf b}}\prod_{L=0}^D f_L(\tilde b_j,j\in\mathcal C_{\Lambda}^{(i)}(L))
\\&=
\prod_{L=0}^D\sum_{\tilde b_j,j\in\mathcal C_{\Lambda}^{(i)}(L)}\!\!\!\!\!f_L(\tilde b_j,j\in\mathcal C_{\Lambda}^{(i)}(L)),
\nonumber
\end{align}
for generic functions $f_L:\{0,1\}^{|\mathcal C_{\Lambda}^{(i)}(L)|}\mapsto\mathbb R$, hence 
by using this fact in the denominator of (\ref{wdsd}) we obtain
\begin{align}
&\pi_{i,0}(\mathbf b)
\\&
=
\prod_{L=0}^D
\frac{\left[
\begin{array}{c}
\mathbb P(\sigma_{i,0}^{(L)}=o_L|b_{j,0}=b_j,\forall j\in\mathcal C_{\Lambda}^{(i)}(L))\\
\times\mathbb P(b_{j,0}=b_j,\forall j\in\mathcal C_{\Lambda}^{(i)}(L))
\end{array}\right]
}
{
\!\!\!\!\!\!\!\!\!\underset{\tilde b_j\in\{0,1\},\forall j\in\mathcal C_{\Lambda}^{(i)}(L)}{\sum}\!
\left[
\begin{array}{c}
\mathbb P(\sigma_{i,0}^{(L)}=o_L|b_{j,0}=\tilde b_j,\forall j\in\mathcal C_{\Lambda}^{(i)}(L))\\\times\mathbb P(b_{j,0}=\tilde b_j,\forall j\in\mathcal C_{\Lambda}^{(i)}(L))
\end{array}\right]
}.
\nonumber
\end{align}
By Bayes' rule we finally obtain
\begin{align}
&\pi_{i,0}(\mathbf b)
=
\prod_{L=0}^D
\mathbb P(b_{j,0}=b_j,\forall j\in\mathcal C_{\Lambda}^{(i)}(L)|\sigma_{i,0}^{(L)}=o_L),
\nonumber
\end{align}
yielding (\ref{eq1}) for $t=0$.

Note that 
\begin{align}
&\mathbb P(b_{j,0}=b_j,\forall j\in\mathcal C_{\Lambda}^{(i)}(L)|\sigma_{i,0}^{(L)}=o_L)
\\&
=
\frac{\left[
\begin{array}{c}
\mathbb P(\sigma_{i,0}^{(L)}=o_L|b_{j,0}=b_j,\forall j\in\mathcal C_{\Lambda}^{(i)}(L))\\
\times\mathbb P(b_{j,0}=b_j,\forall j\in\mathcal C_{\Lambda}^{(i)}(L))
\end{array}\right]
}
{
\!\!\!\!\!\!\!\!\!\underset{\tilde b_j\in\{0,1\},\forall j\in\mathcal C_{\Lambda}^{(i)}(L)}{\sum}\!
\left[
\begin{array}{c}
\mathbb P(\sigma_{i,0}^{(L)}=o_L|b_{j,0}=\tilde b_j,\forall j\in\mathcal C_{\Lambda}^{(i)}(L))\\\times\mathbb P(b_{j,0}=\tilde b_j,\forall j\in\mathcal C_{\Lambda}^{(i)}(L))
\end{array}\right]
}.
\nonumber
\end{align}
Then, using the definition of $\sigma_{i,0}^{(L)}$ in (\ref{sigmadef2}),
and noticing that it is a function of $b_{j,0},\forall j\in\mathcal C_{\Lambda}^{(i)}(L)$,
 we obtain
\begin{align}
&\mathbb P(b_{j,0}=b_j,\forall j\in\mathcal C_{\Lambda}^{(i)}(L)|\sigma_{i,0}^{(L)}=o_L)
\\&
=
\frac{
\chi\left(\underset{j\in\mathcal C_{\Lambda}^{(i)}(L)}{\sum}b_j=o_L\right)\underset{j\in\mathcal C_{\Lambda}^{(i)}(L)}{\prod}\mathbb P(b_{j,0}=b_j)
}
{
\!\!\!\!\!\!\!\!\!\underset{\tilde b_j\in\{0,1\},\forall j\in\mathcal C_{\Lambda}^{(i)}(L)}{\sum}\!
\chi\left(\underset{j\in\mathcal C_{\Lambda}^{(i)}(L)}{\sum}\tilde b_j=o_L\right)\underset{j\in\mathcal C_{\Lambda}^{(i)}(L)}{\prod}\mathbb P(b_{j,0}=\tilde b_j)
}.
\nonumber
\end{align}
Note that, if $\mathbf b$ is such that 
$\underset{j\in\mathcal C_{\Lambda}^{(i)}(L)}{\sum}b_j\neq o_L$,
then $\mathbb P(b_{j,0}=b_j,\forall j\in\mathcal C_{\Lambda}^{(i)}(L)|\sigma_{i,0}^{(L)}=o_L)=0$ as in 
(\ref{eq2}). Conversely, if 
$\underset{j\in\mathcal C_{\Lambda}^{(i)}(L)}{\sum}b_j=o_L$,
for all vectors $\tilde{\mathbf b}$ such that 
$\underset{j\in\mathcal C_{\Lambda}^{(i)}(L)}{\sum}\tilde b_j=o_L$
we have that 
\begin{align}
\underset{j\in\mathcal C_{\Lambda}^{(i)}(L)}{\prod}\mathbb P(b_{j,0}= b_j)=
\underset{j\in\mathcal C_{\Lambda}^{(i)}(L)}{\prod}\mathbb P(b_{j,0}=\tilde b_j),
\end{align}
since $b_{j,0}$ are identically distributed. Thus it follows that 
\begin{align}
&\mathbb P(b_{j,0}=b_j,\forall j\in\mathcal C_{\Lambda}^{(i)}(L)|\sigma_{i,0}^{(L)}=o_L)
\\&
=
\frac{
\chi\left(\underset{j\in\mathcal C_{\Lambda}^{(i)}(L)}{\sum}b_j=o_L\right)
}
{
\!\!\!\!\!\!\!\!\!\underset{\tilde b_j\in\{0,1\},\forall j\in\mathcal C_{\Lambda}^{(i)}(L)}{\sum}\!
\chi\left(\underset{j\in\mathcal C_{\Lambda}^{(i)}(L)}{\sum}\tilde b_j=o_L\right)
}
\label{x1}
\\&
=
\chi\left(\sum_{j\in\mathcal C_{\Lambda}^{(i)}(L)}b_j=o_L\right)
\frac{o_L!\left(|\mathcal C_{\Lambda}^{(i)}(L)|-o_L\right)!}{|\mathcal C_{\Lambda}^{(i)}(L)|!},
\label{x2}
\end{align}
since there are $\left(\begin{array}{c}|\mathcal C_{\Lambda}^{(i)}(L)|\\o_L\end{array}\right)$
 possible combinations of $\{\tilde b_{j},j\in\mathcal C_{\Lambda}^{(i)}(L)\}$
 such that $\underset{j\in\mathcal C_{\Lambda}^{(i)}(L)}{\sum}\tilde b_j=o_L$.
This proves (\ref{eq2}) for $t=0$.

Now, let $t\geq 1$ and assume (\ref{eq1}) and (\ref{eq2}) hold for $t-1$.
We show that they hold at time $t$ as well.
We have
\begin{align}
&\pi_{i,t}(\mathbf b)=\mathbb P(b_{j,t}=b_j,\forall j|\boldsymbol{\sigma}_{i,\tau},\tau\leq t)
\\&
=\frac{
\mathbb P(b_{j,t}=b_j,\forall j,\boldsymbol{\sigma}_{i,t}=(o_0,o_1,\dots,o_D)|\boldsymbol{\sigma}_{i,\tau},\tau\leq t-1)
}
{
\sum_{\tilde{\mathbf b}}\mathbb P(b_{j,t}=\tilde b_j,\forall j,\boldsymbol{\sigma}_{i,t}=(o_0,o_1,\dots,o_D)|\boldsymbol{\sigma}_{i,\tau},\tau\leq t-1)
}
\nonumber
\\&
=\frac{
\left[
\begin{array}{c}
\mathbb P(\boldsymbol{\sigma}_{i,t}=(o_0,o_1,\dots,o_D)|b_{j,t}=b_j,\forall j)
\\\times\mathbb P(b_{j,t}=b_j,\forall j|\boldsymbol{\sigma}_{i,\tau},\tau\leq t-1)
\end{array}\right]
}
{
\sum_{\tilde{\mathbf b}}
\left[
\begin{array}{c}
\mathbb P(\boldsymbol{\sigma}_{i,t}=(o_0,o_1,\dots,o_D)|b_{j,t}=\tilde b_j,\forall j)
\\\times\mathbb P(b_{j,t}=\tilde b_j,\forall j|\boldsymbol{\sigma}_{i,\tau},\tau\leq t-1)
\end{array}\right]
},
\nonumber
\end{align}
where we have used Bayes' rule. Using the fact that 
$\sigma_{i,t}^{(L)}$ in (\ref{sigmadef2})
is a function of $b_{j,t},\forall j\in\mathcal C_{\Lambda}^{(i)}(L)$, we then obtain
\begin{align}
\label{stepx}
&\pi_{i,t}(\mathbf b)
=\frac{
\left[
\begin{array}{c}
\prod_{L=0}^D\chi\left(\underset{j\in\mathcal C_{\Lambda}^{(i)}(L)}{\sum}b_j=o_L\right)
\\\times\mathbb P(b_{j,t}=b_j,\forall j|\boldsymbol{\sigma}_{i,\tau},\tau\leq t-1)
\end{array}\right]
}
{
\sum_{\tilde{\mathbf b}}
\left[
\begin{array}{c}
\prod_{L=0}^D\chi\left(\underset{j\in\mathcal C_{\Lambda}^{(i)}(L)}{\sum}\tilde b_j=o_L\right)
\\\times\mathbb P(b_{j,t}=\tilde b_j,\forall j|\boldsymbol{\sigma}_{i,\tau},\tau\leq t-1)
\end{array}\right]
}.
\end{align}

Note that, since $\{b_{j,t},t\geq 0\}$ is a Markov chain, we obtain
\begin{align}
&\mathbb P(b_{j,t}=b_j,\forall j|\boldsymbol{\sigma}_{i,\tau},\tau\leq t-1)
\\&
=
\sum_{\hat{\mathbf b}}\mathbb P(b_{j,t}=b_j|b_{j,t-1}=\hat b_j,\forall j)\pi_{i,t-1}(\hat{\mathbf b})
\\&
=
\sum_{\hat{\mathbf b}}\pi_{i,t-1}(\hat{\mathbf b})
\prod_{L=0}^D
\prod_{j\in\mathcal C_{\Lambda}^{(i)}(L)}\mathbb P(b_{j,t}=b_j|b_{j,t-1}=\hat b_j),
\label{laststep}
\end{align}
where
\begin{align}
\pi_{i,t-1}(\hat{\mathbf b})=\mathbb P(b_{j,t-1}=\hat b_j,\forall j|\boldsymbol{\sigma}_{i,\tau},\tau\leq t-1).
\end{align}
In the last step of (\ref{laststep}), we have used the fact that 
$\{b_{j,t}\}$ are statistically independent across cells, and that 
$\{\mathcal C_{\Lambda}^{(i)}(L),L=0,1,\dots, D\}$ define a partition of $\{1,2,\dots, N_C\}$.

Now, using the induction hypothesis, 
we can express $\pi_{i,t-1}$ using (\ref{eq1}), hence
\begin{align}
&\mathbb P(b_{j,t}=b_j,\forall j|\boldsymbol{\sigma}_{i,\tau},\tau\leq t-1)
\\&=
\prod_{L=0}^D
\sum_{\hat b_j,j\in\mathcal C_{\Lambda}^{(i)}(L)}\prod_{ j\in\mathcal C_{\Lambda}^{(i)}(L)}\mathbb P(b_{j,t}=b_j|b_{j,t-1}=\hat b_j)
\\&
\times
\mathbb P(b_{j,t-1}=\hat b_j,\forall j\in\mathcal C_{\Lambda}^{(i)}(L)|\sigma_{i,t-1}^{(L)}=\hat o_L),
\end{align}
where we used (\ref{fLprop}).
Then, using (\ref{eq2})
to express $\mathbb P(b_{j,t-1}=\hat b_j,\forall j\in\mathcal C_{\Lambda}^{(i)}(L)|\sigma_{i,t-1}^{(L)}=\hat o_L)$ 
and substituting the resulting expression in (\ref{stepx}), we obtain
\begin{align}
\label{lasgtpit}
&\pi_{i,t}(\mathbf b)
\\&
{=}
\prod_{L=0}^D
\frac{
\left[\!\!\!
\begin{array}{c}
\underset{\hat b_j,j\in\mathcal C_{\Lambda}^{(i)}(L)}{\sum}
\chi\left(\underset{j\in\mathcal C_{\Lambda}^{(i)}(L)}{\sum}b_j=o_L,
\!\!\!\!
\underset{j\in\mathcal C_{\Lambda}^{(i)}(L)}{\sum}
\hat b_j=\hat o_L\right)
\\\times
\prod_{ j\in\mathcal C_{\Lambda}^{(i)}(L)}\mathbb P(b_{j,t}=b_j|b_{j,t-1}=\hat b_j)
\end{array}\!\!\!\right]
}
{
\left[\!\!\!
\begin{array}{c}
\underset{\hat b_j,\tilde b_j,j\in\mathcal C_{\Lambda}^{(i)}(L)}{\sum}
\chi\left(\underset{j\in\mathcal C_{\Lambda}^{(i)}(L)}{\sum}\tilde b_j=o_L,\!\!\!
\underset{j\in\mathcal C_{\Lambda}^{(i)}(L)}{\sum}
\hat b_j=\hat o_L\right)
\\\times\prod_{ j\in\mathcal C_{\Lambda}^{(i)}(L)}\mathbb P(b_{j,t}=\tilde b_j|b_{j,t-1}=\hat b_j)
\end{array}\!\!\!\right]
}.
\nonumber
\end{align}

Note that, for any $\mathbf b$ and $\tilde{\mathbf b}$,
and for any permutation $P(\cdot):\mathcal C_{\Lambda}^{(i)}(L)\mapsto\mathcal C_{\Lambda}^{(i)}(L)$ of the elements in the set $\mathcal C_{\Lambda}^{(i)}(L)$,
we have
\begin{align}
&\prod_{j\in\mathcal C_{\Lambda}^{(i)}(L)}\mathbb P(b_{j,t}=b_j|b_{j,t-1}=\hat b_j)
\\&
=
\prod_{j\in\mathcal C_{\Lambda}^{(i)}(L)}\mathbb P(b_{j,t}=b_{P(j)}|b_{j,t-1}=\hat b_{P(j)}),
\end{align}
since $b_{j,t}$ are statistically identical across cells.

Thus,
for any $\mathbf b$ and $\tilde{\mathbf b}$ such that 
\begin{align}
\label{fghfd}
\underset{j\in\mathcal C_{\Lambda}^{(i)}(L)}{\sum}b_j=\underset{j\in\mathcal C_{\Lambda}^{(i)}(L)}{\sum}\tilde b_j=o_L,
\end{align}
we obtain
\begin{align}
&\prod_{j\in\mathcal C_{\Lambda}^{(i)}(L)}\mathbb P(b_{j,t}=b_j|b_{j,t-1}=\hat b_j)
\\&
=
\prod_{j\in\mathcal C_{\Lambda}^{(i)}(L)}\mathbb P(b_{j,t}=\tilde b_{j}|b_{j,t-1}=\hat b_{P(j)}),
\end{align}
where $P(\cdot)$ is a proper permutation which maps $b_j$ to $\tilde b_j,\forall j\in\mathcal C_{\Lambda}^{(i)}(L)$.
The existence of such permutation is guaranteed by the condition (\ref{fghfd}),
since $\{b_j,j\in\mathcal C_{\Lambda}^{(i)}(L)\}$ and $\{\tilde b_j,j\in\mathcal C_{\Lambda}^{(i)}(L)\}$ have the same number of zero and non-zero elements.

Therefore, for any $\mathbf b$ and $\tilde{\mathbf b}$ satisfying (\ref{fghfd}),
\begin{align*}
&\underset{\hat b_j,j\in\mathcal C_{\Lambda}^{(i)}(L)}{\sum}
\!\!\!\!\chi\left(\sum_{j\in\mathcal C_{\Lambda}^{(i)}(L)}\!\!\!\!\hat b_j{=}\hat o_L\right)\!\!\!
\prod_{ j\in\mathcal C_{\Lambda}^{(i)}(L)}\!\!\!\mathbb P(b_{j,t}=\tilde b_j|b_{j,t-1}=\hat b_j)
\\&
=
\!\!\!\!\!\!\underset{\hat b_j,j\in\mathcal C_{\Lambda}^{(i)}(L)}{\sum}\!\!
\!\!\!\!\chi\left(\sum_{j\in\mathcal C_{\Lambda}^{(i)}(L)}\!\!\!\!\hat b_{P(j)}{=}\hat o_L\right)\!\!\!
\prod_{ j\in\mathcal C_{\Lambda}^{(i)}(L)}\!\!\!\!\!\!\mathbb P(b_{j,t}{=}\tilde b_j|b_{j,t-1}{=}\hat b_{P(j)})
\\&
=\!\!\!\!\!\!
\underset{\hat b_j,j\in\mathcal C_{\Lambda}^{(i)}(L)}{\sum}
\!\!\!\!\chi\left(\sum_{j\in\mathcal C_{\Lambda}^{(i)}(L)}\!\!\!\!\hat b_j{=}\hat o_L\right)\!\!\!
\prod_{ j\in\mathcal C_{\Lambda}^{(i)}(L)}\!\!\!\mathbb P(b_{j,t}=b_j|b_{j,t-1}=\hat b_j),
\end{align*}
where in the last step we have used the fact that 
the sum over 
$\{\hat b_j,j\in\mathcal C_{\Lambda}^{(i)}(L)\}$
covers the same set of elements as
the sum over the set with permuted entries,
$\{\hat b_{P(j)},j\in\mathcal C_{\Lambda}^{(i)}(L)\}$.

Finally, substituting in (\ref{lasgtpit}) we obtain
\begin{align}
&\pi_{i,t}(\mathbf b)
=
\prod_{L=0}^D
\frac{
\chi\left(\underset{j\in\mathcal C_{\Lambda}^{(i)}(L)}{\sum}b_j=o_L\right)
}
{
\underset{\tilde b_j,j\in\mathcal C_{\Lambda}^{(i)}(L)}{\sum}
\chi\left(\underset{j\in\mathcal C_{\Lambda}^{(i)}(L)}{\sum}\tilde b_j=o_L\right)
}.
\nonumber
\end{align}
Note that this expression implies that 
$\pi_{i,t}(\mathbf b)$ is only a function of  $\boldsymbol{\sigma}_{i,t}$ but is independent of $\boldsymbol{\sigma}_{i,\tau},\tau\leq t-1$;
additionally, $\{b_j,j\in\mathcal C_{\Lambda}^{(i)}(L_1)\}$ is statistically independent of 
$\{b_j,j\in\mathcal C_{\Lambda}^{(i)}(L_2)\}$ for $L_1\neq L_2$, given $\boldsymbol{\sigma}_{i,t}$.
Thus, (\ref{eq1}) follows, where
$\mathbb P(b_{j,t}=b_j,\forall j\in\mathcal C_{\Lambda}^{(i)}(L)|\sigma_{i,t}^{(L)}=o_L)$ is given as in 
(\ref{x1}).
Finally, (\ref{eq2}) follows from (\ref{x1})-(\ref{x2}).

The induction step, and the Theorem, are thus proved.
\end{proof}
}
{
}

\bibliographystyle{IEEEtran}
\bibliography{IEEEabrv,bibliography} 

\end{document}